\documentclass[12pt,reqno,twoside]{amsart}
\usepackage{amsmath,amssymb,amsfonts,amsthm,indentfirst,setspace}
\usepackage{xcolor}\usepackage{ulem}

\newtheorem{theorem}{Theorem}[section]

\newtheorem{corollary}{Corollary}[section]
\newtheorem{remark}{Remark}[section]
\newcommand{\be}{\begin{equation}}
\newcommand{\ee}{\end{equation}}
\newcommand{\bea}{\begin{eqnarray}}
\newcommand{\eea}{\end{eqnarray}}
\newcommand{\eeas}{\end{eqnarray*}}
\newcommand{\beas}{\begin{eqnarray*}}

\begin{document}

\title{ALMOST PSEUDO-RICCI SYMMETRIC SPACETIME SOLUTIONS IN $F(R)$-GRAVITY}
\author{Avik De \and Tee-How Loo}
\address{A. De\\
Department of Mathematical and Actuarial Sciences\\
Universiti Tunku Abdul Rahman\\
Jalan Sungai Long\\
43000 Cheras\\
Malaysia}
\email{de.math@gmail.com}
\address{T. H. Loo\\
Institute of Mathematical Sciences\\
University of Malaya\\
50603 Kuala Lumpur\\
Malaysia}
\email{looth@um.edu.my}
\date{}

\thanks{The authors were supported in part by the FRGS research grant (Grant No. FP137-2019A)}

\begin{abstract}
The objective of the present paper is to study 4-dimensional almost pseudo Ricci symmetric perfect fluid spacetimes  $(APRS)_4$. We show that a Robertson-Walker spacetime is $(APRS)_4$ and vice versa under certain condition imposed on its scale factor. Some popular toy models of $F(R)$-gravity are also studied under the current setting and various energy conditions are investigated.
\end{abstract}

\maketitle


\section{Introduction}

A semi-Riemannian manifold is said to be locally symmetric if the curvature tensor is parallel in the sense that $\nabla_r R_{ijkl}=0$. 
After E. Cartan completely classified all locally symmetric Riemannian spaces   \cite{cartan}, the notion has been weakened by many authors to different extent such as conformally symmetric manifolds  \cite{mccgupta}, recurrent manifolds  \cite{walker}, conformally recurrent manifolds \cite{adatimiza}, pseudo symmetric manifolds \cite{mcc87}, weakly symmetric manifolds \cite{tamsbin},	weakly Ricci symmetric manifolds \cite{tamsbin93} etc. 
Later Chaki \cite{chaki} introduced a pseudo-Ricci symmetric manifold $(PRS)_n$ as a non-flat semi-Riemannian manifold whose Ricci tensor $R_{ij}$ of type $(0, 2)$ is not identically zero and satisfies the condition
\[
\nabla_iR_{jk}=2A_iR_{jk}+A_jR_{ki}+A_kR_{ij},
\]
where $A_i$ is a non-zero 1-form. 

When $A_i$ vanishes, we obtain a Ricci-symmetric manifold, that is, the Ricci tensor satisfies 
$\nabla_iR_{jk}=0$.
 With its several interaction with general relativistic works, Chaki and Kawaguchi motivated to generalize further the concept and introduced an almost pseudo-Ricci symmetric manifold $(APRS)_n$ \cite{aaa}. 
A non-flat semi-Riemannian manifold is called an $(APRS)_n$ if its Ricci tensor $R_{ij}$
is not identically zero and satisfies the condition 
\be \label{aprs}
\nabla_iR_{jk}=(A_i+B_i)R_{jk}+A_jR_{ki}+A_kR_{ij},
\ee
where $A_i$ and $B_i$ are two 1-forms called associated 1-forms.
If $A_i=B_i$, an $(APRS)_n$ reduces to a pseudo-Ricci symmetric
manifold, making $(PRS)_n$ a particular case of $(APRS)_n$. Several authors studied $(APRS)_n$ in different settings, 
see \cite{avik}, \cite{ucd1}, \cite{ucd2} etc and the references therein.  

A Lorentzian manifold is said to be an almost pseudo Ricci symmetric spacetime \cite{deghosh} if the Ricci tensor $R_{ij}$ satisfies (\ref{aprs}). 
Bektas et al. recently investigated a perfect fluid $(APRS)_4$ spacetime solution of Einstein's field equations without cosmological constant, where the four-velocity vector field $u_i=A_i$ \cite{ozen} which motivated us towards the present study. 

The Einstein's field equations (EFE)
\[
R_{ij}-\frac{R}{2}g_{ij}=\kappa^2 T_{ij},
\] 
are unable to explain the late time inflation of universe without assuming the existence of some yet undetected components abbreviated as dark energy. 
Here 
		$\kappa^2=8\pi G$, 
		$G$ is the Newton's gravitational constant, 
		$R$ is the Ricci scalar  and 
		$T_{ij}$ is the energy momentum tensor describing the matter content of the spacetime. 

Some researchers started believing that EFE might break down at a large scale and tried to modify it to get some higher order field equations of gravity.
One of these modified gravity theories is obtained by replacing the Ricci scalar $R$ in the Einstein-Hilbert action with an arbitrary function $F(R)$ of $R$. 
The $F(R)$ theories, despite being the simplest generalization, are believed to be the unique one among higher-order gravity theories which can avoid the long known and fatal Ostrogradski instability \cite{sotiriou}. 
Of course the viability of such functions are constrained by several observational data and scalar-tensor theoretical results. 
However, we have certain functions which could be able to explain the whole cosmic history of the universe starting from the early accelerated expansion, decelerated expansion and the late time accelerated expansion. 
Additionally we can always propose some phenomenological assumption  about the form of the function $F(R)$ and later verify its validity from the present viability criteria.

The matter content in the gravitational field equations is more often than not assumed to be a perfect fluid continuum having density, pressure and possessing dynamical and kinematical quantities like velocity, acceleration, vorticity, shear and expansion, in which case the energy momentum tensor $T_{ij}$ of type $(0,2)$ is given by 
\begin{align}\label{eqn:T}
T_{ij}=pg_{ij}+(\sigma + p)u_iu_j,
\end{align}
where $\sigma $ and $p$ are the energy density and the isotropic pressure respectively, 
the timelike $u^i$ is the velocity vector field of the fluid. 
$\theta=\nabla_iu^i$ is termed as the expansion scalar of the fluid, 
$\dot u^l=u^i\nabla_iu^l$ is the acceleration vector, 
\[
s_{jl}=\frac{1}{2}h_j^ih_l^r(\nabla_iu_r+\nabla_ru_i)-\frac{\theta}{3}h_{jl}
\]
is the shear tensor, where $h_{jl}=g_{jl}+u_j u_l$ is the orthogonal projector. 
	Physically, the expansion measures the rate of change of the volume of a fluid element per unit volume, and the shear tensor measures the shearing deformation of a fluid element. The vorticity which measures the local rotation of the fluid is given by 
$$\Omega_{jl}=\frac{1}{2}h_j^ih_l^r(\nabla_iu_r-\nabla_ru_i),$$ 
mathematically it is equivalent to the curl of the velocity vector field $u^i$ of the fluid.

In addition, $p$ and $\sigma $ are related by an equation of state governing the particular sort of perfect fluid under consideration. 
In general, this is an equation of the form $p=p(\sigma )$. 
In this case, the perfect fluid is called isentropic. 
Moreover, if $p=\sigma$, then the perfect fluid is termed as stiff matter. 
A stiff matter equation of   state was   first introduced by Zeldovich in \cite{z1} and used in his cosmological model in which the  primordial universe is assumed to be a cold gas of baryons \cite{z2}. 
The stiff matter era preceded the radiation era with $p=\frac{\sigma}{3}$, the dust matter era with $p=0$ and followed by the dark matter era with $p=-\sigma$ \cite{c2}.

The present paper is organized as follows: After the introduction, in Section 2 an equivalent condition for a Robertson-Walker spacetime to be an $(APRS)_4$ is deduced. Next, we study almost pseudo Ricci symmetric spacetimes with constant Ricci scalar which satisfy $F(R)$-gravity equations. In the next section we discuss the energy conditions in such a setting, followed by some toy models of $F(R)$-gravity investigated in $(APRS)_4$ with constant $R$. We close the study with a discussion. 


\section{Robertson-Walker spacetime as an $(APRS)_4$}

The current favored model of our universe is spatially homogeneous and isotropic or mathematically speaking, a warped product $\mathbb{R}\times_{a(t)}M^3$, popularly known as a Robertson-Walker (RW) spacetime. Here the manifold $M^3$, in general, is a space form of curvature $-1,\, 0$ or $1$ but recent observational data convince us of a spatially flat universe $k=0$ case. The function $a(t)$ is called the scale factor of the universe and $\frac{\dot{a}}{a}=H$ the Hubble parameter. In this section we show that the RW spacetime, is almost pseudo Ricci symmetric if and only if the scale factor $a(t)$ satisfies certain conditions. 

The line element and the Ricci scalar in a spatially flat RW spacetime are respectively given by
\begin{align*}
ds^2=-dt^2+a^2(t)\left(dr^2+r^2d\theta^2+r^2\sin^2\theta d\phi^2\right),
\end{align*}
\begin{align*}
R=6 \frac{a\ddot a+\dot a^2}{a^2}.
\end{align*}

The Ricci tensor takes the form
\begin{align}\label{eqn:7a}
R_{jl}=(P-Q)u_ju_l+Pg_{jl}=-Qu_ju_l+Ph_{jl}
\end{align}
where 
\begin{align}\label{eqn:2a}
P=\frac{a\ddot a+2\dot a^2}{a^2}, \quad Q=3\frac{\ddot a}{a} 
\end{align}
and  $u^i=(\partial_t)^i$ is the four-velocity of the fluid with $u^ju_j=-1$ and  
\begin{align}\label{eqn:3}
\nabla_ju_l=&\frac{\dot a}ah_{jl}.   
\end{align}
It is clear from (\ref{eqn:2a}) that 
\begin{align}\label{eqn:7b}
\nabla_iP=&-u_i\dot P, \quad \nabla_iQ=-u_i\dot Q.
\end{align} 
Taking covariant derivative on (\ref{eqn:7a}), with the help of (\ref{eqn:3})--(\ref{eqn:7b}) we obtain
\begin{align}\label{eqn:9a}
\nabla_iR_{jl} 
=&(\nabla_iP-\nabla_iQ)u_ju_l+(P-Q)\{\nabla_iu_ju_l+\nabla_iu_lu_j\}+\nabla_iPg_{jl} \notag \\
=&\dot Qu_iu_ju_l+(P-Q)\frac{\dot a}a\{h_{ij}u_l+h_{il}u_j\}-\dot Pu_ih_{jl}.	  
\end{align}

Now let us further assume that it is an $(APRS)_4$ spacetime. By (\ref{aprs}) and (\ref{eqn:7a}), we obtain
\begin{align}\label{eqn:9b}
\nabla_iR_{jl}
=& -Q\{(A_i+B_i)u_ju_l+A_ju_iu_l+A_lu_iu_j \} \notag\\
 & +P\{(A_i+B_i)h_{jl}+A_jh_{il}+A_lh_{ij} \} .
\end{align}
By comparing (\ref{eqn:9a})--(\ref{eqn:9b}), we have
\begin{align}\label{eqn:10}	
\dot Qu_iu_ju_l&+Q\{(A_i+B_i)u_ju_l+A_ju_iu_l+A_lu_iu_j\} \notag\\
&=X_ih_{jl}+Y_jh_{il}+Y_lh_{ij} 
\end{align}
where $X_i=P(A_i+B_i)+\dot Pu_i$ and  
$Y_j=PA_j-(P-Q)({\dot a}/a)u_j$.  
Let $h^{jl}=g^{jl}+u^ju^l$. Then
\[
h_{ij}h^{jl}=h_i^l=\delta_i^l+u_iu^l.
\]
Transvecting (\ref{eqn:10}) with $h^{jl}$, we have
\begin{align}\label{eqn:12}
3X_i+2Y_jh_{i}^j=0. 
\end{align}
Transvecting with $u^i$, we have $3X_iu^i=0$. Similarly, we have $Y_iu^i=0$.
It follows that (\ref{eqn:12}) becomes
\begin{align*}
3X_i+2Y_i=0. 
\end{align*}

Similarly, we obtain $2X_i+3Y_i=0.$ Solving these two equations gives
$X_i=Y_i=0$ or 
\begin{align}\label{eqn:20}
P(A_i+B_i)=-\dot Pu_i; \quad 		
PA_i=(P-Q)\frac{\dot a}au_i.
\end{align}
If $P=0$, then $Q=0$ by (\ref{eqn:20}).
This case is infeasible due to physical constraints. 
Hence we assume that $P\neq 0$. It follows from (\ref{eqn:20}) that $P-Q\neq0$.   
Further, the equation (\ref{eqn:10}) is simplified as 
\[\left(\dot Q-Q\frac{\dot P}P+2Q\frac{P-Q}{P}\frac{\dot a}a\right)u_iu_ju_l=0.\]	
So
\begin{align}\label{eqn:30}
P\dot Q-Q{\dot P}-2Q(Q-P)\frac{\dot a}a=0.
\end{align}

A RW spacetime is an $(APRS)_4$ spacetime if and only if the conditions
(\ref{eqn:20})--(\ref{eqn:30}) are satisfied.
Due to physical considerations, we assume  that $Q\neq0$.
Hence 
\begin{align*}
0=&\frac{1}{Q^2a^2}\left(P\dot Q-Q\dot P-2Q(Q-P)\frac{\dot a}a\right)
			=\frac{1}{Q^2a^2}\frac d{dt}\left(\frac{Q-P}{Qa^2}\right).
\end{align*}	
Since $P-Q\neq0$, we have 
\begin{align}\label{a''-1}
\frac{a\ddot a-\dot a^2}{a^3\ddot a}=-\frac1\epsilon
\end{align}
where $\epsilon\neq 0$ is a constant.
This is a second order DE. We shall transform it into a first order DE.
To do this, we first transform this equation into
\[
0=\frac{a^3}{2\dot a}\left(2\dot a\ddot a\frac{1+\epsilon^{-1}a^2}{a^2}-\dot a^2\frac{2\dot a}{a^3}
\right)
=\frac1\epsilon\frac{a^3}{2\dot a}\frac{d}{dt}\left(\dot a^2\frac{\epsilon+a^2}{a^2}\right).
\]
Hence we obtain 
\begin{align}\label{a'a'-1}
\frac{\dot a^2}{a^2}(\epsilon+a^2)=\psi
\end{align}
where $\psi\neq0$ is a constant. Furthermore, (\ref{a''-1})--(\ref{a'a'-1}) imply that 
\begin{align}\label{a''}
\frac{\dot a^2}{a^2}=\frac{\psi}{\epsilon+a^2}; \quad 
\frac{\ddot a}a
=\frac{\psi}{\epsilon+a^2}\frac{\epsilon}{\epsilon+a^2}.
\end{align}
By applying (\ref{a''}) to (\ref{eqn:2a}), we compute
\begin{align*}
p=&\frac{\psi}{\epsilon+a^2}\left\{\frac{\epsilon}{\epsilon+a^2}+2\right\} \\
p'=&\frac{-2a\dot a}{\epsilon+a^2}\frac{\psi}{\epsilon+a^2}\left\{\frac{2\epsilon}{\epsilon+a^2}+2\right\}	\\
p-q=&\frac{\psi}{\epsilon+a^2}\frac{2a^2}{\epsilon+a^2}.
\end{align*}
Noticing also that (\ref{a'a'-1}) gives
\begin{align*}
\frac{\dot a}{a}\sqrt{|\epsilon+a^2|}=\epsilon_1
\end{align*}
where $\epsilon_1$ is a constant with $\epsilon_1^2=|\psi|$.
By substituting these equations into (\ref{eqn:20}), we obtain
\begin{align*}
A_i=\frac{2\epsilon_1a^2}{3\epsilon+2a^2}\frac{1}{\sqrt{|\epsilon+a^2|}}u_i; \quad 
B_i=\frac{2\epsilon_1a^2}{3\epsilon+2a^2}\frac{1}{\sqrt{|\epsilon+a^2|}}\frac{3\epsilon+a^2}{\epsilon+a^2}u_i.
\end{align*}
Finally we consider three cases:
\begin{enumerate}
	\item[(a)] When $\epsilon>0$. Let $\epsilon=\epsilon_0^2$ with $\epsilon_0>0$. Then 
	\[
  \frac{\dot a}{a}\sqrt{\epsilon_0^2+a^2}=\epsilon_1
	.
   \]
	Furthermore, by the integration formula, we obtain
		\[
		\sqrt{\epsilon_0^2+a^2}-\epsilon_0\ln(\epsilon_0+\sqrt{\epsilon_0^2+a^2})+\epsilon_0\ln a=\epsilon_1 t+\epsilon_2.
		\]
	\item[(b)] When $\epsilon<0$ and $\epsilon+a^2<0$. Let $\epsilon=-\epsilon_0^2$ with $\epsilon_0>0$. 
	Then $|\epsilon+a^2|=\epsilon_0^2-a^2$ and so 
	\[
  \frac{\dot a}{a}\sqrt{\epsilon_0^2-a^2}=\epsilon_1 
	.
   \]
	Furthermore, by the integration formula, we obtain
		\[
		\sqrt{\epsilon_0^2-a^2}-\epsilon_0\ln(\epsilon_0+\sqrt{\epsilon_0^2-a^2})+\epsilon_0\ln a=\epsilon_1 t+\epsilon_2.
		\]	
	\item[(c)] When $\epsilon<0$ and  $\epsilon+a^2>0$. Let $\epsilon=-\epsilon_0^2$ with $\epsilon_0>0$. 
	Then $|\epsilon+a^2|=a^2-\epsilon_0^2$ and so 
	\[
  \frac{\dot a}{a}\sqrt{a^2-\epsilon_0^2}=\epsilon_1 
	.
   \]
	Furthermore, by the integration formula, we obtain
		\[
		\sqrt{a^2-\epsilon_0^2}-\epsilon_0\sec^{-1}\frac{a}{\epsilon_0}=\epsilon_1 t+\epsilon_2.
		\]	
\end{enumerate}
Thus we have the following:
\begin{theorem}
A spatially flat RW spacetime is an $(APRS)_4$ spacetime if and only if either 
\begin{enumerate}
	\item[(a)] the scale factor $a(t)$ satisfies the implicit equation:
	\[
	\sqrt{\epsilon_0^2+a^2}-\epsilon_0\ln(\epsilon_0+\sqrt{\epsilon_0^2+a^2})+\epsilon_0\ln a=\epsilon_1 t+\epsilon_2
	\]
	and the associated $1$-forms are given respectively by
	\begin{align*}
	A_i=\frac{2\epsilon_1a^2}{3\epsilon_0^2+2a^2}\frac{1}{\sqrt{\epsilon_0^2+a^2}}u_i; \quad 
	B_i=\frac{2\epsilon_1a^2}{3\epsilon_0^2+2a^2}\frac{3\epsilon_0^2+a^2}{(\epsilon_0^2+a^2)^{3/2}}u_i
	\end{align*}
	where $\epsilon_0>0$, $\epsilon_1$ and $\epsilon_2$ are constants; or 
	\item[(b)] the scale factor $a(t)$ satisfies the implicit equation:
	\[
	\sqrt{\epsilon_0^2-a^2}-\epsilon_0\ln(\epsilon_0+\sqrt{\epsilon_0^2-a^2})+\epsilon_0\ln a=\epsilon_1 t+\epsilon_2
	\]
	and the associated $1$-forms are given respectively by
	\begin{align*}
	A_i=\frac{2\epsilon_1a^2}{2a^2-3\epsilon_0^2}\frac{1}{\sqrt{\epsilon_0^2-a^2}}u_i; \quad 
	B_i=\frac{2\epsilon_1a^2}{2a^2-3\epsilon_0^2}\frac{3\epsilon_0^2-a^2}{(\epsilon_0^2-a^2)^{3/2}}u_i
	\end{align*}
	where $\epsilon_0>0$, $\epsilon_1$  and $\epsilon_2$ are constants; or 
	\item[(c)] the scale factor $a(t)$ satisfies the implicit equation:
	\[
	\sqrt{a^2-\epsilon_0^2}-\epsilon_0\sec^{-1}\frac{a}{\epsilon_0}=\epsilon_1 t+\epsilon_2
	\]
	and the associated $1$-forms are given respectively by
	\begin{align*}
	A_i=\frac{2\epsilon_1a^2}{2a^2-3\epsilon_0^2}\frac{1}{\sqrt{a^2-\epsilon_0^2}}u_i; \quad 
	B_i=\frac{2\epsilon_1a^2}{2a^2-3\epsilon_0^2}\frac{a^2-3\epsilon_0^2}{(a^2-\epsilon_0^2)^{3/2}}u_i
	\end{align*}
	where $\epsilon_0>0$, $\epsilon_1$  and $\epsilon_2$ are constants. 
\end{enumerate}
\end{theorem}


\section{$(APRS)_4$ satisfying $F(R)$-gravity}

In an $(APRS)_4$ the covariant derivative of the Ricci tensor satisfies (\ref{aprs}). Hence we have
\[
\nabla_iR_{jk}-\nabla_kR_{ij}=B_iR_{jk}-B_kR_{ij},
\]
which on contraction over $j$ and $k$ gives us 
\be\label{dr}
\nabla_iR=2RB_i-2R_{ij}B^j.
\ee
If we consider a constant Ricci scalar $R$, we get from (\ref{dr}), 
\be \label{barb}
RB_i=R_{ij}B^j.
\ee

We consider a modified Einstein-Hilbert action term 
\[S=\frac{1}{\kappa^2}\int F(R) \sqrt{-g}d^4x +\int L_m\sqrt{-g}d^4x,\]
where 
		$F(R)$ is an arbitrary function of the Ricci scalar $R$, 
		$L_m$ is the matter Lagrangian density, and we define the stress-energy tensor of matter as 
\[T_{ij}=-\frac{2}{\sqrt{-g}}\frac{\delta(\sqrt{-g}L_m)}{\delta g^{ij}}.\] 
By 	varying 	the 	action $S$ of the gravitational field with respect to the metric tensor components $g^{ij}$ and using the least action principle we obtain the $f(R)$-gravity field equations
\be \label{FR}	
F_R(R)R_{ij}-\frac{1}{2}F(R)g_{ij}+(g_{ij}\Box-\nabla_i\nabla_j)F_R(R)=\kappa^2T_{ij},
\ee
where $\Box$ represents the d'Alembertian operator and $F_R=\frac{\partial F(R)}{\partial R}$. 
Einstein's field equations can be reawakened by putting $F(R)=R$. 

The trace of (\ref{FR}) gives 
\beas 3\Box F_R(R)+RF_R(R)-2F(R)=\kappa^2 T,\eeas
which we can rewrite as 
\[
\Box F_R(R)=\frac{\partial V^{\text{eff}}}{\partial F_R(R)}.
\]
On the critical points, the effective potential $V^{\text{eff}}$ has a maximum (or minimum), so that 
\beas \Box F_R(R_{CP})=0\eeas 
and 
\beas 2F(R_{\text{CP}})-R_{\text{CP}}F_R(R_{\text{CP}})=-\kappa^2T.\eeas
Here, $R_{\text{CP}}$ is the curvature at the critical point.
 For example, in absence of matter, i.e., $T=0$, one has the de Sitter critical point associated with a constant Ricci scalar $R_{\text{dS}}$. For a constant Ricci scalar, we can express the above field equations (\ref{FR}) as follows:
\be \label{fr}
R_{ij}-\frac{R}{2}g_{ij}=\frac{\kappa^2}{F_R(R)}T^{\text{eff}}_{ij},
\ee
where 
$$T^{\text{eff}}_{ij}=T_{ij}+\frac{F(R)-RF_R(R)}{2\kappa^2}g_{ij}.$$ 
Remembering the term $\kappa^2=8\pi G$, the quantity $G^{\text{eff}}=\frac{G}{F_R(R)}$ can be regarded as the effective gravitational coupling strength in analogy to what is done in Brans-Dicke type scalar-tensor gravity theories and further the positivity of $G^{\text{eff}}$ (equivalent to the requirement that the graviton is not a ghost) imposes that the effective scalar degree of freedom or the scalaron term $f_R(R) > 0$.

In (\cite{CMM1},\cite{CMM2},\cite{CMM3}) cosmological perfect fluid case is considered in various gravity theories. Motivated by these studies we consider a perfect fluid spacetime that satisfies (\ref{fr}). Hence using (\ref{eqn:T}) we have 
\be \label{a}
R_{ij}=\frac{\kappa^2(p+\sigma)}{F_R(R)}u_iu_j+\frac{2\kappa^2p+F(R)}{2F_R(R)}g_{ij}.
\ee
It follows that 
\be \label{a1}
R_{ij}u^j=\frac{F(R)-2 \kappa^2\sigma}{2F_R(R)} u_i.
\ee
Throughout this study we consider a perfect fluid $(APRS)_4$ spacetime solution of $f(R)$-gravity equations where the 
four-velocity vector $u_i=B_i$, so that (\ref{barb}) gives $R_{ij}u^j=Ru_i$. Hence, we conclude that 
\[
R=\frac{F(R)-2 \kappa^2\sigma}{2F_R(R)}
\]
or 
\be \label{a3} 
 \sigma=\frac{F(R)-2RF_R(R)}{2\kappa^2}.
\ee
On the other hand, the trace equation of (\ref{a}) is given by 
\[R=\frac{-\kappa^2\sigma+3\kappa^2p+2F(R)}{F_R(R)}.\]
This, together with (\ref{a1}),  gives
\be \label{a4} 
p=-\frac{F(R)}{2\kappa^2}.
\ee
This leads to our first result:
\begin{theorem}\label{pfthm}
In a perfect fluid $(APRS)_4$ spacetime with constant $R$ satisfying $F(R)$-gravity; 
if the four-velocity vector $u^i=B^i$, then its isotropic pressure $p$ and energy density $\sigma$ are given by 
$p=-\frac{F(R)}{2\kappa^2}$ and $\sigma=\frac{F(R)-2RF_R(R)}{2\kappa^2}$. 
Moreover, both the pressure and density are constant in this special scenario.
\end{theorem}

\begin{corollary}\label{thm:2.2}
A vacuum $(APRS)_4$ spacetime solution with constant $R$ 
and $u_i=B_i$ 
is not viable in the $F(R)$-gravity theory. 
\end{corollary}
\begin{proof}
For vacuum case, $T_{ij}=0$ and $p=\sigma=0$. It follows from (\ref{a4}) that $F(R)=0$.
\end{proof}

\begin{remark}
In general relativity, $F(R)=R$, so the perfect fluid represents a stiff matter $p=\sigma=-\frac{R}{2\kappa^2}$.
\end{remark}
\begin{theorem} \label{thm:2.3}
The matter content in a perfect fluid $(APRS)_4$ spacetime with constant $R$ and $u_i=B_i$ satisfying $F(R)$-gravity obeys the simple barotropic equation of state $p=\omega \sigma$ if and only if 
$F(R)=\lambda R^{(1+\omega)/2\omega}$.
\end{theorem}

\begin{proof}
Suppose $p=\omega\sigma$. It follow from (\ref{a3})-(\ref{a4}) that 
\[
2\omega RF_R(R)=(1+\omega)F(R).
\]
Solving this equation gives $F(R)=\lambda R^{(1+\omega)/2\omega}$ where $\lambda$ is a constant.
The converse is trivial.
\end{proof}
\begin{remark}Corresponding to the different states of cosmic evolution of the universe we can conclude:
\begin{itemize}
\item The perfect fluid denotes dark matter ($\omega=-1$) if $F(R)$ is a constant function of $R$ or alternately if the spacetime is scalar flat.
\item The perfect fluid denotes stiff matter ($\omega =1$) if  $F(R)$ is a constant multiple of $R$. 
\item The perfect fluid denotes radiation ($\omega=1/3$) if $F(R)$ is a constant multiple of $R^2$.
\item The perfect fluid cannot represent a dust era for any viable $F(R)$. 
\end{itemize}
\end{remark}

\begin{theorem} \label{thm:2.4}
In a perfect fluid $(APRS)_4$ spacetime with constant $R$ satisfying $F(R)$-gravity;
if the four-velocity vector $u^i=B^i$, then the fluid either has vanishing expansion scalar and acceleration vector or represents a dark matter. 
\end{theorem}

\begin{proof}
By Theorem~\ref{pfthm}, the pressure and density of the perfect fluid are constant. 
Using the conservation of energy $\nabla^iT_{ij}=0$ and  (\ref{eqn:T}), we obtain
\be \label{eqn:divT}
0=(p+\sigma)\{\nabla^iu_i u_j+u_i\nabla^iu_j\}. \ee
Since $u_ju^j=-1$, $\nabla^iu_j u^j=0$. 
Contracting by $u^j$ on (\ref{eqn:divT}) gives
\[
0=-(p+\sigma)\nabla^iu_i. 
\]
Furthermore, we also have 
\[
0=(p+\sigma)u_i\nabla^iu_j. 
\]
Hence, either $p+\sigma=0$ or $\nabla^iu_i =u_i\nabla^iu_j=0$.
\end{proof}
Since a conservative vector field is always irrotational, we get the vorticity of the perfect fluid is zero.


\begin{theorem} \label{thm:2.5}
If a perfect fluid spacetime with constant $R$ satisfying $F(R)$-gravity obeys the timelike convergence condition, then $\sigma\geq\frac{F(R)}{2k^2}$.
\end{theorem}
\begin{proof}
$u_i$ is timelike, hence timelike convergence implies that 
$$R_{ij}u^iu^j\geq 0.$$ 
As discussed earlier, $F_R(R)>0$ to ensure attractive gravity. Therefore, from (\ref{a1}) we obtain the result. 
\end{proof}

\section{Energy conditions in an $(APRS)_4$}
Energy conditions are coordinate-invariant restrictions on the (effective) energy-momentum tensor which is useful when we explore the possibility of variety of matter sources, not necessarily only a perfect fluid continuum, which satisfy the Einstein's field equations or the modified theories of gravity and preserve the idea that energy should be positive. There are several energy conditions; some of which are obsolete these days like the trace energy condition, some are weaker and included in the other. But in general, the idea is to contract the energy momentum tensor with arbitrary timelike or lightlike vectors to produce some scalar fields. We use Theorem \ref{pfthm} to deduce some relevant energy conditions for our study. 

\begin{itemize}
\item Null energy condition \textbf{(NEC)}: 
The weakest of all, it states that $T^{\text{eff}}_{ij}l^il^j\geq0$ for all null vector $l^i$, or in non-technical terms it says that the energy density of the fields contributing to $T^{\text{eff}}_{ij},$ as measured in a natural way by any observer is never negative.
In the present context this condition gives us $RF_R(R)\leq 0.$ 

\item Weak energy condition \textbf{(WEC)}: 
It states that, $T^{\text{eff}}_{ij}t^i t^j\geq 0,$ for all timelike vectors $t^i$. 
This also implies, by continuity, the NEC. 
In the present context, considering the timelike vector $u^i$ we obtain
$RF_R(R)\leq 0$.

\item Dominant energy condition \textbf{(DEC)}: 
It states that matter flows along timelike or null world lines. 
Mathematically, $T^{\text{eff}}_{ij}t^i t^j\geq 0$ for any timelike $t^i$ together with $T^{\text{eff}}_{ij}t^i$ is not spacelike, either null or timelike. 
By continuity the property should also hold true for any null vector $l^i$. 
In the present context we obtain, $RF_R(R)\leq 0.$

\item Strong energy condition \textbf{(SEC)}:
It states that $T^{\text{eff}}_{ij}t^i t^j\geq \frac{1}{2}T^i_i t^j t_j,$ for all timelike vectors $t^i$ which after some calculations reduces to $RF_R(R)\leq 0$.
\end{itemize}

Since, $F_R(R)>0$ and we considered $R\neq 0$, we can conclude $R< 0$ from the energy conditions.
 
Finally, we should cite \cite{energy} where Curiel elaborately discussed about various energy conditions; their consequences in terms of formation of singularities, thermodynamics, black hole theories etc, and the violations by some classical fields for further insight on energy conditions and their importance.


\section{Analysis of some toy models of $F(R)$-gravity in $(APRS)_4$}

Here we consider two of the earliest toy models of $F(R)$-gravity theories to analyse our results in a perfect fluid $(APRS)_4$ with constant Ricci scalar setting and with the four-velocity vector $u^i=B^i$.  

\textbf{Case I:} $F(R)=R-\frac{\mu^4}{R}.$\\
This first model was considered by Carroll et al \cite{carroll} to explain the late-time acceleration. The equation (\ref{a}) in this case reduces to 
\[
R_{ij}=\frac{\kappa^2(p+\sigma)}{1+\mu^4/R^2}u_iu_j+\frac{2\kappa^2p+R-\mu^4/R}{2(1+\mu^4/R^2)}g_{ij},
\]
with $p=-\frac{R-\mu^4/{R}}{2\kappa^2}$ and $\sigma=-\frac{R+2\mu^4/R^2}{2k^2}$.

\textbf{Case II:} $F(R)=R+\alpha R^2$.\\
The most representative model of $R^2$ cosmology is this so called Starobinsky \cite{starobinsky} model with
the help of which inflation can be explained without a need for a scalar field. 
The Equation (\ref{a}) reduces to 
\[
R_{ij}=\frac{\kappa^2(p+\sigma)}{1+2\alpha R}u_iu_j+\frac{2\kappa^2p+R+\alpha R^2}{2(1+2\alpha R)}g_{ij},
\]
with $p=-\frac{R+\alpha R^2}{2\kappa^2}$ and $\sigma=-\frac{R+3\alpha R^2}{2\kappa^2}$.


\section{Discussion}
In the present study we investigate an almost pseudo Ricci symmetric spacetime $(APRS)_4$ in the modified gravity scenario. The current model of the universe, namely, the Robertson-Walker spacetime is shown to be almost pseudo Ricci symmetric under certain condition. We consider an $(APRS)_4$ with constant Ricci scalar satisfying the $f(R)$-gravity where the matter content of the gravity theory represents a perfect fluid with the four-velocity vector $u_i=B_i$ and find the expressions for the pressure and energy density. The fluid in this case is seen to either represents a dark matter or its expansion scalar, acceleration vector and vorticity vanish. Several energy conditions are studied in this setting, some toy models of $f(R)$-gravity is discussed. 
   

\section{Acknowledgement}The authors are grateful to the referees for their valuable suggestions towards the improvement of the paper.

\end{document}